\documentclass[nofootinbib,twocolumn,pra,letterpaper,longbibliography]{revtex4-2}
\usepackage{latexsym,amsmath,amssymb,amsfonts,graphicx,color,amsthm}
\usepackage{enumerate}
\usepackage{braket}
\usepackage{adjustbox}
\usepackage{footnote}
\usepackage[dvipsnames]{xcolor}
\usepackage{tipa}
\usepackage{textcomp}
\usepackage{fontenc}
\usepackage{mathrsfs}
\usepackage{color}
\usepackage{colortbl}
\usepackage{graphics,graphicx}
\usepackage{txfonts}
\usepackage{lipsum, babel}
\usepackage{bbm}

\usepackage[unicode=true,pdfusetitle, bookmarks=true,bookmarksnumbered=false,bookmarksopen=false, breaklinks=false,pdfborder={0 0 0},backref=false,colorlinks=false] {hyperref}
\hypersetup{ colorlinks,linkcolor=myurlcolor,citecolor=myurlcolor,urlcolor=myurlcolor}

\definecolor{myurlcolor}{rgb}{0,0,0.7}

\usepackage{float}
\usepackage{hyperref}

\newcommand{\cC}{\mathcal{C}}
\newcommand{\cD}{\mathcal{D}}

\newcommand{\cH}{\mathcal{H}}
\newcommand{\cI}{\mathcal{I}}
\newcommand{\cJ}{\mathcal{J}}
\newcommand{\cK}{\mathcal{K}}
\newcommand{\cL}{\mathcal{L}}
\newcommand{\cM}{\mathcal{M}}
\newcommand{\cN}{\mathcal{N}}

\newcommand{\cP}{\mathcal{P}}

\newcommand{\cS}{\mathcal{S}}

\newcommand{\Id}{\mathbbm{1}}
\newcommand{\tr}{\text{Tr}}

\newtheorem{theorem}{Theorem}

\newtheorem{lemma}{Lemma}
\newtheorem{corollary}[theorem]{Corollary}
\newtheorem{definition}{Definition}

\begin{document}

\title{Relating CP divisibility of dynamical maps with compatibility of channels}
\author{Arindam Mitra$^{1,2,3,4}$}
\email{20003292@iitb.ac.in}
\email{arindammitra143@gmail.com}

\author{ Debashis Saha$^{5}$}
\email{saha@iisertvm.ac.in}

\author{ Samyadeb Bhattacharya$^{6,7}$}
\email{samyadeb.b@iiit.ac.in}

\author{ A. S. Majumdar$^{8}$}
\email{archan@bose.res.in}

\affiliation{$^1$Department of Physics, Indian Institute of Technology Bombay, Mumbai 400076, India.\\
$^2$Centre of Excellence in Quantum Information, Computation, Science, and Technology, Indian Institute of Technology Bombay, Powai, Mumbai 400076, India\\
$^3$Optics and Quantum Information Group, The Institute of Mathematical Sciences, C. I. T. Campus, Taramani, Chennai 600113, India.\\
$^4$Homi Bhabha National Institute, Training School Complex, Anushaktinagar, Mumbai 400094, India.}

\affiliation{$^5$School of Physics, Indian Institute of Science Education and Research Thiruvananthapuram, Kerala 695551, India
}

\affiliation{$^{6}$Centre for Quantum Science and Technology, International Institute of Information Technology-Hyderabad, Gachibowli, Hyderabad-500032, Telangana, India.\\
$^7$Center for Security, Theory and Algorithmic Research, International Institute of Information Technology-Hyderabad, Gachibowli, Hyderabad-500032, Telangana, India.}

\affiliation{$^{8}$S. N. Bose National Centre for Basic Sciences, Block JD, Sector III, Salt Lake, Kolkata 700 106, India}

\date{\today}

\begin{abstract}
 The role of CP-indivisibility and incompatibility as valuable resources for various information-theoretic tasks is widely acknowledged. This study delves into the intricate relationship between CP-divisibility and channel compatibility. Our investigation focuses on the behaviour of incompatibility robustness of quantum channels for a pair of generic dynamical maps. We show that the incompatibility robustness of channels is monotonically non-increasing for a pair of generic CP-divisible dynamical maps. Further, 
 our explicit study of the behaviour of incompatibility robustness with time for some specific dynamical maps reveals non-monotonic behaviour in the CP-indivisible regime. Additionally, we propose a measure of  CP-indivisibility based on the incompatibility robustness of quantum channels. Our investigation provides valuable insights into the nature of quantum dynamical maps and their relevance in information-theoretic applications. 
\end{abstract}

\maketitle
\section{Introduction}
Incompatibility is one of the main features of quantum mechanics that makes it different from classical mechanics \cite{Heino-review}. A set of devices is said to be compatible if those devices can be simultaneously implemented on a quantum system. Otherwise, the set is incompatible. These devices can be measurements, channels, instruments etc. Incompatibility is a resource in several information theoretic tasks and is necessary to demonstrate non-classical advantage in such tasks. For example, measurement incompatibility is necessary and sufficient to demonstrate quantum steering \cite{Brunner-comp-bell}.  Incompatibility of measurements is also necessary to demonstrate Bell inequality violation \cite{Barnett-comp-bell,Brunner-comp-bell} and any quantum advantage in communication tasks \cite{Saha2023}.
Measurement incompatibility also provides advantage in some state discrimination tasks \cite{Skrzypczyk-comp-state-disc}. Recently, it has been shown that incompatibility of channels and measurement-channel incompatibility both provide advantages in quantum state discrimination tasks \cite{Mori-comp-chan-state-disc}. 

A general quantum evolution is described by a completely positive trace preserving (CPTP) dynamical map \cite{breuer,alicki,lindblad,gorini,RHPreview,BLPreview,Vegareview}. This representation has wide application, since it is almost impossible to keep a quantum system truly isolated. When considering a Markovian evolution \citep{RHPreview,BLPreview}, it is necessary to take the quantum system weakly interacting with much larger stationary environment, and hence, the reduced dynamics of the system can be considered to be memoryless, leading to one-way information flow from the quantum system to the bath degrees of freedom. Therefore, the quantum features of a system subjected to such dynamics vanishes gradually with time \cite{RHPreview,BLPreview,Vegareview}. However, in practical situations like in an experiment, the coupling between the system and environmental degrees of freedom may not always be sufficiently weak. Moreover, the concerning environment can very well be finite or non-stationary. These situations may lead to the signature of non-Markovian information backflow \cite{RHP,blp1,bellomo,arend,ban1,ban2,Bhattacharya17,samya2,Bhattacharya20,Maity20,BBhattacharya21}. Though  quantum non-markovianity has been associated with varied physical attributes \citep{modi1,modi2,RHP,blp1}, the focus of this work is based solely on indivisibility of the dynamics exhibiting information backflow from the environment to the system \citep{RHP}. 

A divisible quantum operation is the one, that can be realised as arbitrary number of completely positive trace preserving (CPTP) maps. In other words, such operations can be divided into arbitrary number of CPTP maps. The precise mathematical definition of such maps is later presented in equation \eqref{Eq:def-divi} for better understanding.
Divisible maps do not exhibit information backflow from environment to the system \citep{blp1,RHP}, and hence, can be understood as Markovian operations. 
The Born-Markov approximation and stationary bath state approximation are imperative to realize such quantum operations, and hence, in the absence of these initial approximations, the dynamics is bound to be CP-indivisible
and  prone to show information backflow \citep{RHPreview}. Adopting this line of reasoning, in this work, we take CP-indivisbile quantum operations as non-Markovian. Note that CP-indivisble operations  are necessary to have information backflow from the environment to the system enabling recovery
of lost information to an extent, and hence, they can be considered resourceful operations in information processing scenarios.  For example, it has been shown that information backflow allows perfect teleportation with mixed states \citep{resource1}, improvement of capacity for long quantum channels \citep{resource2} and efficient work extraction from an Otto cycle \citep{resource3}.

From the above discussions, it is clear that (in-) compatibility and CP-(in-) divisibility can play the role of resources in various information-theoretic and thermodynamic tasks. 
There are several resources in quantum theory that provide advantage in  information-theoretic or thermodynamic tasks \cite{Chitambar-review-resource}. Entanglement 
 \cite{RevModPhys.81.865}, coherence \cite{PhysRevA.107.012221}, non-locality \cite{PhysRevLett.97.120405}, contextuality \cite{howard2014contextuality,RevModPhys.94.045007}, and incompatibility \cite{PhysRevLett.124.120401} are examples of some widely studied resources. Evidently therefore, exploring the interplay amongst different resources forms an important avenue of research. For example, it is well-known that coherence can be measured with entanglement \cite{Adesso-coh-entang}.
 The relation between incompatibility and Bell non-locality \cite{Barnett-comp-bell,Brunner-comp-bell} as well as the relation between incompatibility and steerability \cite{Brunner-comp-bell} are well-known. Furthermore, it has been recently shown that superposition and entanglement are equivalent concepts in any physical theory \cite{Lami-sup-entang}. 
 
 The motivation for the 
 present work is
to explore the connection between the two resources of (in-)compatibility and CP-(in-)divisibility, which have been hitherto investigated separately
in the literature. It is known that both incompatibility and CP-indivisibility are resources for different information-theoretic tasks.
Moreover, to the best of our knowledge,  the notion of compatibility has been considered only for devices in the static backdrop. 
In Ref.\cite{Duarte2022}, the authors have done a qualitative study regarding the relation between CP-divisibility and
incompatibility. However, they didn’t consider dynamical maps that involve time, and neither consider the robustness measure
of incompatibility to draw quantitative connection between the above-said resources.   On the other hand, 
in this work we introduce and characterize the notion of compatibility of dynamical maps, incorporating their evolution in time.
 Through our present analysis, we  characterize CP-(in-)divisibility with respect to the (in-)compatibility of channels. We study the behaviour of incompatibility robustness of quantum channels for some examples of dynamical maps. We further present an example where
    the non-Markovian advantange manifested in terms of CP-indivisibility and information backflow is clearly seen to act as a quantum resource
    in the task of teleportation.  Moreover, we define a measure of CP-indivisibility based on incompatibility robustness of channels. 

The rest of the paper is organized as follows.  In Sec. \ref{sec:prelims}, we provide definitions of various quantities required for the subsequent analysis. Our main results are presented from Sec. \ref{sec:main} onwards.
In Sec. \ref{sec:cp-div-comp-chan}, we show that the incompatibility robustness of quantum channels for a pair of CP-divisible dynamical maps is monotonically non-increasing with respect to time. In Sec. \ref{sec:rel-with-comp-meas}, we show that for any pair of dynamical maps, incompatibility robustness of measurements is upper bounded by incompatibility of channels for an arbitrary time. In Sec \ref{sec:comp-DM}, we discuss the notion of compatibility of dynamical maps and its connection to channel compatibility. In Sec. \ref{sec:Behav-RoI-chan-examp-DM}, we study the behaviour of incompatibility robustness of quantum channels for certain specific dynamical maps and show its non-monotonic behaviour in a CP-indivisible case.  In Sec. \ref{sec:telport-cp-in-divisible-incomp}, we discuss the usefulness of CP-indivisibility in the context of quantum teleportation and compare the behaviour of teleportation fidelity with the incompatibility robustness w.r.t. time. In Sec. \ref{sec:cp-in-meas}, we propose a measure of CP-indivisibility based on incompatibility of channels. Finally, in Sec. \ref{sec:conc}, we 
present our concluding remarks.

\section{Preliminaries}\label{sec:prelims}

\subsection{Compatibility of measurements}
A measurement $M$ acting on the Hilbert space $\cH$ is a set of positive semi-definite matrices i.e., $M=\{M(x)\}_{x\in\Omega_M}$ such that $\sum_{x\in\Omega_M}M(x)=\Id_{\cH}$ where $\Id_{\cH}$ is the identity matrix on the Hilbert space $\cH$ and $\Omega_M$ is known as the outcome set of $M$. The set of all measurements acting on Hilbert space $\cH$ and with outcome set $\Omega$ is denoted by $\mathbbm{M}(\cH,\Omega)$. A set of measurements $\cM=\{M_i\}^n_{i=1}$ is said to be compatible if there exists a joint measurement $M=\{M(j_1,\ldots, j_n)\}\in\mathbbm{M}(\cH,\Omega_M)$ with $\Omega_M=\Omega_{M_1}\times\ldots\times\Omega_{M_n}$ such that $M_i(j_i)=\sum_{\{j_k\}\setminus j_i}M(j_1,\ldots, j_n)$ for all $j_i\in\Omega_{M_i}$ for all $i\in\{1,\ldots n\}$ where sum over $\{j_k\}\setminus j_i$ denotes the sum over all $j_k$s for all $k$s except for $k=i$. Otherwise, the set is incompatible \cite{Heino-review, kiukas2023joint}.

A  measure of incompatibility of quantum measurements is the incompatibility robustness of quantum measurements,  defined below. The incompatibility robustness of two quantum measurements $M_1\in\mathbbm{M}(\cH,\Omega_1)$ and $M_2\in\mathbbm{M}(\cH,\Omega_2)$ can be defined as:
\begin{equation}\label{Def:ROI-channels}
\begin{split}
R_{M}(M_1,M_2)=\min~  & \left. r\right. \\
\text{s.t.} ~ & \left. \frac{M_1(i_1)+r\tilde{M}_1(i_1)}{1+r}=\sum_{i_2}M(i_1.i_2)\right. \\
& \left. \frac{M_2(i_2)+r\tilde{M}_2(i_2)}{1+r}=\sum_{i_1}M(i_1.i_2)\right. \\
& \left. M \in\mathbbm{M}(\cH,\Omega_1\times\Omega_2)\right. \\
& \left. \tilde{M}_i\in\mathbbm{M}(\cH,\Omega_i) \quad i = 1,2. \right.
\end{split}
\end{equation}
Here, $\tilde{M}_i$s are arbitrary noise measurements and the optmization is over all variables, other than the given pair of measurements $(M_1,M_2)$. We call the set of all values of $r$ that satisfies the above equalities for different noise measurements, \emph{the compatibility range}. Clearly, the incompatibility robustness is the minimum of all values of $r$ that belongs to the compatibility range.

\subsection{Compatibility of quantum channels}\label{Prelim:comp-chan}
A quantum channel $\Gamma:\cL(\cH)\rightarrow\cL(\cK)$ is a CPTP linear map where $\cL(\cH)$ is the bounded linear operator on the Hilbert space $\cH$ and $\cL(\cK)$ is the bounded linear operator on the Hilbert space $\cK$. We denote the set of all quantum channels from $\cL(\cH)$ to $\cL(\cK)$ as $\mathbbm{Ch}(\cH,\cK)$. We also denote the composition (also known as concatenation) of two quantum channels $\Phi_1:\cL(\cH)\rightarrow\cL(\bar{\cK})$ and $\Phi_2:\cL(\bar{\cK})\rightarrow\cL(\cK)$ as $\Phi_2\circ\Phi_1$ where $\bar{K}$ is another Hilbert space. For two quantum channels $\Gamma_1\in\mathbbm{Ch}(\cH,\cK_1)$ and $\Gamma_2\in\mathbbm{Ch}(\cH,\cK_2)$, if there exists another quantum channel $\Theta\in\mathbbm{Ch}(\cK_1,\cK_2)$ such that $\Gamma_2=(\Theta\circ\Gamma_1)$ then we say that $\Gamma_2$ is a post-processing of $\Gamma_1$ and we denote it as $\Gamma_2\preceq\Gamma_1$. Let $\Phi\in\mathbbm{Ch}(\cH,\cK)$ be a quantum channel. Then the $\Phi^*:\cL(\cK)\rightarrow\cL(\cH)$ is called as a dual map of $\Phi$ if $\tr[\Phi(T)X]=\tr[T\Phi^*(X)]$ holds for all $T\in\cL(\cH)$ and $X\in\cL(\cK)$. Clearly, $\Phi^*$ is the action of $\Phi$ in Heisenberg picture. As, $\Phi$ is CP trace preserving, $\Phi^*$ is CP unital. Now, consider a measurement $M=\{M(x)\}_{x\in\Omega_M}\in\mathbbm{M}(\Omega_M, \cK)$ and a channel $\Lambda\in\mathbbm{Ch}(\cH,\cK)$. If $\Lambda^*$ is applied on the measurement $M$, the resulting measurement is $\Lambda^*(M)=\{\Lambda^*(M(x))\}_{x\in\Omega_M}\in\mathbbm{M}(\Omega_M, \cH)$. Implementation of $M$ on an arbitrary quantum system after implementation of $\Lambda$, is equivalent to implementation of $\Lambda^*(M(x))$ before implementation of $\Lambda$ on that quantum system. 

We now discuss a special type of channel that maps any input state to a fixed output state. These channels are called completely depolarising (CD) channels \cite{Heino-incomp-chan} that may completely erase the information of input states.   If $\Upsilon_{\eta}(\rho)=\eta$ for all input states $\rho$ then $\Upsilon_{\eta}$ is a completely depolarising channel. We denote the set of all completely depolarising channels from $\cL(\cH)$ to $\cL(\cK)$ as $\mathbbm{Ch}^{\cC\cD}(\cH,\cK)$. We will use this type of channels in the later sections. Now, suppose $\Upsilon_{\eta}\in\mathbbm{Ch}^{\cC\cD}(\cH,\cK)$ is an arbitrary completely depolarising channel and $\Lambda_1\in\mathbbm{Ch}(\cK,\cK^{\prime})$ is an arbitrary quantum channel. Then $(\Lambda_1\circ\Upsilon_{\eta})$ is also a completely depolarising channel \cite{Heino-incomp-chan}.

Below, we define the compatibility of quantum channels.
\begin{definition}
Two quantum channels $\Lambda_1:\cL(\cH)\rightarrow\cL(\cK_1)$ and $\Lambda_2:\cL(\cH)\rightarrow\cL(\cK_2)$ are compatible if there exists a quantum channel  $\Lambda:\cL(\cH)\rightarrow\cL(\cK_1\otimes\cK_2)$ such that for all $T\in\cL(\cH)$ 
\begin{equation}
\Lambda_1(T)=\tr_{\cK_2}\Lambda(T);~\Lambda_2(T)=\tr_{\cK_1}\Lambda(T)\label{Eq:chan-comp}. 
\end{equation}
Otherwise, $\Lambda_1$ and $\Lambda_2$ are incompatible \cite{Heino-incomp-chan} \label{Def:Comp-QC}.
\end{definition}
The quantum channel $\Lambda$ in Definition \ref{Def:Comp-QC} is also known as the joint quantum channel. Equation \ref{Eq:chan-comp} can be re-written using short-hand notation as 
\begin{equation}
\Lambda_1=\tr_{\cK_2}\Lambda;~\Lambda_2=\tr_{\cK_1}\Lambda. 
\end{equation}
We will use the short-hand notations throughout the paper. 

Suppose $\bar{\Gamma}_1\preceq\Gamma_1$ and $\bar{\Gamma}_2\preceq\Gamma_2$. Then it is proved in Ref. \cite[Proposition 3]{Heino-incomp-chan} that $\bar{\Gamma}_1$ and $\bar{\Gamma}_2$ are compatible if $\Gamma_1$ and $\Gamma_2$ are compatible. We will use this result in the proof of Theorem \ref{Th:rel-comp-chan-CP-div-DM} and Theorem \ref{Th:rel-comp-chan-CP-div-DM-fix-noise}.

A  measure of incompatibility of quantum channels is the incompatibility robustness of quantum channels that is defined below \cite{Mori-comp-chan-state-disc}. The incompatibility robustness of two quantum channels $\Phi_1: \cL(\cH) \to \cL(\cK_1)$ and $\Phi_2: \cL(\cH) \to \cL(\cK_2)$ can be defined as:
\begin{equation}\label{Def:ROI-channels}
\begin{split}
R_{C}(\Phi_1,\Phi_2)=\min~  & \left. r\right. \\
\text{s.t.} ~ & \left. \frac{\Phi_1+r\tilde{\Phi}_1}{1+r}=\tr_{\cK_2}{\Psi}\right. \\
& \left. \frac{\Phi_2+r\tilde{\Phi}_2}{1+r}=\tr_{\cK_1}{\Psi}\right. \\
& \left. \Psi \in\mathbbm{Ch}(\cH,\cK_1 \otimes \cK_2)\right. \\
& \left. \tilde{\Phi}_i\in\mathbbm{Ch}(\cH,\cK_i) \quad i = 1,2. \right.
\end{split}
\end{equation}

Here, $\tilde{\Phi}_i$s are arbitrary noise channels, and the optimization is over all variables, other than the given pair of channels $(\Phi_1,\Phi_2)$. We call the set of all values of $r$ that satisfies the above equalities, \emph{the compatibility range}. Clearly, the incompatibility robustness is the minimum of all values of $r$ that belongs to the compatibility range. Note that the definition of incompatibility robustness of quantum channels does not directly guarantee that the compatibility range is continuous. Such a statement is proved in Lemma \ref{App-lemma} below. It is known that for any two given channels, $R_C$ is upper bounded by $1$ \cite[Example 2]{Heino-incomp-chan}. Now, following Ref.\cite{Heino-layers},  broadcasting quantum channels can
be defined.

\begin{definition}
A channel $\Gamma:\cL(\cH)\rightarrow\cL(\cH_1\otimes\cH_2)$  with $\cH_1=\cH_2$ is known as a broadcasting quantum channel \cite{Heino-layers}.
\end{definition}
This definition will be used in later sections.
\subsection{Dynamical maps and CP-divisibility}

A dynamical map is a family of CPTP linear maps $\{\Lambda_{t,t_0}:\cL(\cH)\rightarrow\cL(\cH)\}_t$ (where $t\geq t_0$). Here, $t$ represents the time and $t_0$ is the fixed initial time. Without loss of generality, we can take $t_0$ to be $0$ and denote $\Lambda_{t,0}$ as $\Lambda_t$. We denote the set of all dynamical maps on $\cL(\cH)$ (i.e., from $\cL(\cH)$ to $\cL(\cH)$) as $\mathbbm{DM}(\cH,\cH)$. Let $\cD_1=\{\Lambda^1_t\}_t\in\mathbbm{DM}(\cH,\cH)$ and $\cD_2=\{\Lambda^2_t\}_t\in\mathbbm{DM}(\cH,\cH)$. Then the convex combination of $\cD_1$ and $\cD_2$ (w.r.t. $p\geq 0$) is defined as $p\cD_1+(1-p)\cD_2:=\{p\Lambda^1_t+(1-p)\Lambda^2_t\}_t$. Now, we provide the definition of a CP-divisible dynamical map below.
\begin{definition}
A dynamical map is $\cD=\{\Lambda_t\}_t$ is called CP divisible if for all $t$ and all $s$, it can be written as 
\begin{equation}
\Lambda_t=V_{t,s}\circ\Lambda_s,~(t\geq s)\label{Eq:def-divi}
\end{equation}
 where $\circ$ denotes the composition of maps and $V_{t,s}:\cL(\cH)\rightarrow\cL(\cH)$ is a CPTP linear map.
\end{definition}

A dynamical map that is not CP-divisible is known as CP-indivisible dynamical map. We refer the readers to Refs.\cite{BLPreview,RHPreview,Vegareview} for more details.

\section{Relating CP-divisibility of dynamical maps with compatibility of quantum channels}\label{sec:main}


\subsection{CP-indivisibility of dynamical maps and incompatibility robustness of quantum channels} \label{sec:cp-div-comp-chan}
In this subsection, we establish a connection of CP-divisibility of dynamical maps with compatibility of quantum channels.

\begin{theorem}
Suppose that the quantum dynamical maps $\cD_1=\{\Lambda^1_t\}_t\in\mathbbm{DM}(\cH,\cH)$ and $\cD_2=\{\Lambda^2_t\}_t\in\mathbbm{DM}(\cH,\cH)$ are both CP-divisible. Then $R_C(\Lambda^1_t,\Lambda^2_t)\geq R_C(\Lambda^1_{t+\delta{t}},\Lambda^2_{t+\delta{t}})$ for any $\delta{t}\geq 0$. \label{Th:rel-comp-chan-CP-div-DM}
\end{theorem}

\begin{proof}
If both $\cD_1$ and $\cD_2$ are CP divisible, then
\begin{align}
\Lambda^1_{t+\delta{t}}&=V^1_{t,t+\delta{t}}\circ\Lambda^1_t;\\
\Lambda^2_{t+\delta{t}}&=V^2_{t,t+\delta{t}}\circ\Lambda^2_t
\end{align}
hold for any $\delta{t}\geq 0$, where $V^1_{t,t+\delta{t}}$ and $V^2_{t,t+\delta{t}}$ are quantum channels.
Now, suppose that $R_C(\Lambda^1_t,\Lambda^2_t)=l$. Therefore, from the definition of incompatibility robustness of quantum channels, it follows that there exist two quantum channels $\bar{\Lambda}^1_t$ and $\bar{\Lambda}^2_t$ such that the channels $\Sigma^1_t=\frac{\Lambda^1_t+l\bar{\Lambda}^1_t}{1+l}$ and $\Sigma^2_t=\frac{\Lambda^2_t+l\bar{\Lambda}^2_t}{1+l}$ are compatible. Now, consider the quantum channels,

\begin{align}
V^1_{t,t+\delta{t}}\circ\Sigma^1_t&=\frac{\Lambda^1_{t+\delta{t}}+l(V^1_{t,t+\delta{t}}\circ\bar{\Lambda}^1_t)}{1+l}\label{Eq:1-Rt>Rt+dt}\\
V^2_{t,t+\delta{t}}\circ\Sigma^2_t&=\frac{\Lambda^2_{t+\delta{t}}+l(V^2_{t,t+\delta{t}}\circ\bar{\Lambda}^2_t)}{1+l}\label{Eq:2-Rt>Rt+dt}.
\end{align}
Next, as mentioned in Sec. \ref{Prelim:comp-chan}, for two pairs of channels $(\Gamma_1, \Gamma_2)$ and $(\bar{\Gamma}_1, \bar{\Gamma}_2)$, if $\bar{\Gamma}_1\preceq\Gamma_1$ and $\bar{\Gamma}_2\preceq\Gamma_2$ hold, then $\bar{\Gamma}_1$ and $\bar{\Gamma}_2$ are compatible if $\Gamma_1$ and $\Gamma_2$ are compatible \cite[Proposition 3]{Heino-incomp-chan}. Therefore, as $\Sigma^1_t$ and $\Sigma^2_t$ are compatible, from Eq. \eqref{Eq:1-Rt>Rt+dt} and Eq. \eqref{Eq:2-Rt>Rt+dt}, we get that the quantum channels $V^1_{t,t+\delta{t}}\circ\Sigma^1_t$ and $V^2_{t,t+\delta{t}}\circ\Sigma^2_t$ are also compatible. Therefore, from the definition of incompatibility robustness of quantum channels, it follows that $R_C(\Lambda^1_t,\Lambda^2_t)=l\geq R_C(\Lambda^1_{t+\delta{t}},\Lambda^2_{t+\delta{t}})$.
\end{proof}

Note that the incompatibility robustness defined in Eq. \eqref{Def:ROI-channels}, has the minimization over all possible noise channels. Incompatibility robustness can be defined w.r.t. only the set of all completely depolarising channels, as well. The incompatibility robustness of two quantum channels $\Phi_1: \cL(\cH) \to \cL(\cK_1)$ and $\Phi_2: \cL(\cH) \to \cL(\cK_2)$ w.r.t. completely depolarising channels can be defined as:
\begin{equation}\label{Def:ROI-channels-cd-noise}
\begin{split}
R^{\cC\cD}_{C}(\Phi_1,\Phi_2)=\min~  & \left. r\right. \\
\text{s.t.} ~ & \left. \frac{\Phi_1+r\Theta_1}{1+r}=\tr_{\cK_2}{\Psi}\right. \\
& \left. \frac{\Phi_2+r\Theta_2}{1+r}=\tr_{\cK_1}{\Psi}\right. \\
& \left. \Psi \in\mathbbm{Ch}(\cH,\cK_1 \otimes \cK_2)\right. \\
& \left. \Theta_i\in\mathbbm{Ch}^{\cC\cD}(\cH,\cK_i) \quad i = 1,2. \right.
\end{split}
\end{equation}

A similar measure for the incompatibility of measurements has been  studied
earlier \cite{Heino-review,Heino-open}. 

\begin{theorem}
Suppose that the quantum dynamical maps $\cD_1=\{\Lambda^1_t\}_t\in\mathbbm{DM}(\cH,\cH)$ and $\cD_2=\{\Lambda^2_t\}_t\in\mathbbm{DM}(\cH,\cH)$ are both CP-divisible. Then $R^{\cC\cD}_C(\Lambda^1_t,\Lambda^2_t)\geq R^{\cC\cD}_C(\Lambda^1_{t+\delta{t}},\Lambda^2_{t+\delta{t}})$ for any $\delta{t}\geq 0$. \label{Th:rel-comp-chan-CP-div-DM-fix-noise}
\end{theorem}

\begin{proof}
If both $\cD_1$ and $\cD_2$ are CP divisible then
\begin{align}
\Lambda^1_{t+\delta{t}}&=V^1_{t,t+\delta{t}}\circ\Lambda^1_t;\\
\Lambda^2_{t+\delta{t}}&=V^2_{t,t+\delta{t}}\circ\Lambda^2_t
\end{align}
hold for any $\delta{t}\geq 0$ where $V^1_{t,t+\delta{t}}$ and $V^2_{t,t+\delta{t}}$ are quantum channels.
Now, suppose that $R^{\cC\cD}_C(\Lambda^1_t,\Lambda^2_t)=l$. Therefore, from the definition of incompatibility robustness of quantum channels w.r.t. only the set of completely depolarising channels, we get that there exist two completely depolarising channels $\Theta_1\in\mathbbm{Ch}^{\cC\cD}(\cH,\cK_1)$ and $\Theta_2\in\mathbbm{Ch}^{\cC\cD}(\cH,\cK_2)$ such that the channels $\Sigma^1_t=\frac{\Lambda^1_t+l\Theta_1}{1+l}$ and $\Sigma^2_t=\frac{\Lambda^2_t+l\Theta_2}{1+l}$ are compatible. Now, consider the quantum channels $V^1_{t,t+\delta{t}}\circ\Sigma^1_t=\frac{\Lambda^1_{t+\delta{t}}+l(V^1_{t,t+\delta{t}}\circ\Theta_1)}{1+l}$ and $V^2_{t,t+\delta{t}}\circ\Sigma^2_t=\frac{\Lambda^2_{t+\delta{t}}+l(V^2_{t,t+\delta{t}}\circ\Theta_2)}{1+l}$. Now, as mentioned in Sec. \ref{Prelim:comp-chan}, for two pairs of channels $(\Gamma_1, \Gamma_2)$ and $(\bar{\Gamma}_1, \bar{\Gamma}_2)$, if $\bar{\Gamma}_1\preceq\Gamma_1$ and $\bar{\Gamma}_2\preceq\Gamma_2$ hold, then $\bar{\Gamma}_1$ and $\bar{\Gamma}_2$ are compatible if $\Gamma_1$ and $\Gamma_2$ are compatible \cite[Proposition 3]{Heino-incomp-chan}. Therefore, as $\Sigma^1_t$ and $\Sigma^2_t$ are compatible, we get that the quantum channels $V^1_{t,t+\delta{t}}\circ\Sigma^1_t=\frac{\Lambda^1_{t+\delta{t}}+l(V^1_{t,t+\delta{t}}\circ\Theta_1)}{1+l}$ and $V^2_{t,t+\delta{t}}\circ\Sigma^2_t=\frac{\Lambda^2_{t+\delta{t}}+l(V^2_{t,t+\delta{t}}\circ\Theta_2)}{1+l}$ are also compatible. Therefore, from the definition of incompatibility robustness of quantum channels w.r.t. only the set of completely depolarising channels and the fact that he channels $(V^1_{t,t+\delta{t}}\circ\Theta_1)$ and $(V^2_{t,t+\delta{t}}\circ\Theta_2)$ are completely depolarising, it follows that $R^{\cC\cD}_C(\Lambda^1_t,\Lambda^2_t)=l\geq R^{\cC\cD}_C(\Lambda^1_{t+\delta{t}},\Lambda^2_{t+\delta{t}})$.

\end{proof}

Note that Theorem \ref{Th:rel-comp-chan-CP-div-DM} and Theorem \ref{Th:rel-comp-chan-CP-div-DM-fix-noise} do not directly imply each other. 

\subsection{CP-divisibility and compatibility of quantum measurements} \label{sec:rel-with-comp-meas}
The relation between compatibility of measurements and CP-divisibility of dynamical maps has been studied in detail in Ref. \cite{Heino-open,kiukas2023joint}. Here, we further study it in the context of our present analysis.
In Ref. \cite{Heino-open}, the authors studied the behaviour of incompatibility robustness of quantum measurements with respect to fixed noise measurements. Here, we show that incompatibility robustness of measurements (with respect to generic noise) for any pair of dynamical maps is upper bounded by incompatibility robustness of channels  for that pair of dynamical maps.

\begin{theorem}
    Consider an arbitrary pair of measurements $\cM=\{M_i\in\mathbbm{M}(\Omega_{M_i},\cH)\}_{i\in{1,2}}$ and an arbitrary pair of dynamical maps $\cD_1=\{\Lambda^1_t\}_t\in\mathbbm{DM}(\cH,\cH)$ and $\cD_2=\{\Lambda^2_t\}_t\in\mathbbm{DM}(\cH,\cH)$. Then
\begin{align}
    \max_{M_1,M_2} R_M(\Lambda^{1*}_t(M_1), \Lambda^{2*}_t(M_2))\leq R_C(\Lambda^1_t, \Lambda^2_t).
\end{align}\label{Th:G-RoIM-CP-div}
\end{theorem}

\begin{proof}
    Let $R_C(\Lambda^1_t, \Lambda^2_t)=l(t)$.  Then there exist two noise channels $\tilde{\Lambda}^1_t$ and $\tilde{\Lambda}^2_t$ such that

    \begin{align}
       \frac{\Lambda^1_t+l(t)\tilde{\Lambda}^1_t}{1+l(t)}=\tr_{\cH_2}\Lambda \label{Eq:L1-comp}\\
       \frac{\Lambda^2_t+l(t)\tilde{\Lambda}^2_t}{1+l(t)}=\tr_{\cH_1}\label{Eq:L2-comp}\Lambda
    \end{align}
    where $\Lambda\in\mathbbm{Ch}(\cH,\cH_1\otimes\cH_2)$ with $\cH_1=\cH_2=\cH$.

Using the definition of dual maps, Eq. \eqref{Eq:L1-comp} and Eq. \eqref{Eq:L2-comp}, we obtain

\begin{align}
    \frac{\Lambda^{1*}_t(M_1(i))+l(t)\tilde{\Lambda}^{1*}_t(M_1(i))}{1+l(t)}=\Lambda^*(M_1(i)\otimes\Id_{\cH_2})\\
    \frac{\Lambda^{2*}_t(M_2(j))+l(t)\tilde{\Lambda}^{2*}_t(M_2(j))}{1+l(t)}=\Lambda^*(\Id_{\cH_1}\otimes M_2(j))
\end{align}
for all $i\in\Omega_{M_1}$ and $j\in\Omega_{M_2}$. Let the measurement $M:=\{M(i,j)=\Lambda^*(M_1(i)\otimes M_2(j))\}$. Then clearly, 
\begin{align}
    \sum_jM(i,j)=\frac{\Lambda^{1*}_t(M_1(i))+l(t)\tilde{\Lambda}^{1*}_t(M_1(i))}{1+l(t)}\\
    \sum_i M(i,j)=\frac{\Lambda^{2*}_t(M_2(j))+l(t)\tilde{\Lambda}^{1*}_t(M_2(j))}{1+l(t)}.
\end{align}

Hence, the measurements $M^{\prime}_1=\{\frac{\Lambda^{1*}_t(M_1(i))+l(t)\tilde{\Lambda}^{1*}_t(M_1(i))}{1+l(t)}\}_{i\in\Omega_{M_1}}$ and $M^{\prime}_2=\{\frac{\Lambda^{2*}_t(M_2(j))+l(t)\tilde{\Lambda}^{2*}_t(M_2(j))}{1+l(t)}\}_{i\in\Omega_{M_2}}$ are compatible. Thus, from the definition of incompatibility robustness of quantum measurements, we get $R_M(\Lambda^{1*}_t(M_1), \Lambda^{2*}_t(M_2))\leq l(t)=R_C(\Lambda^1_t, \Lambda^2_t)$.

\end{proof}

Clearly, from Theorem \ref{Th:rel-comp-chan-CP-div-DM} and Theorem \ref{Th:G-RoIM-CP-div}, it follows that the upper bound of $R_M(\Lambda^{1*}_t(M_1), \Lambda^{2*}_t(M_2))$ is monotonically non-increasing if both $\cD_1$ and $\cD_2$ are CP-divisible.

\subsection{CP-divisibility and compatibility of dynamical maps}\label{sec:comp-DM}

First, we define broadcasting quantum dynamical maps, as follows.

\begin{definition}
A broadcasting quantum dynamical map is a family of CPTP linear maps $\{\Lambda_{t,t_0}:\cL(\cH)\rightarrow\cL(\cH_1\otimes\cH_2)\}_t$ (where $t\geq t_0$) with $\cH_1=\cH_2$.
\end{definition}
We denote the set of all broadcasting dynamical maps on $\cL(\cH)$ (i.e., from $\cL(\cH)$ to $\cL(\cH_1\otimes\cH_2)$) with $\cH_1=\cH_2$ as $\mathbbm{BD}(\cH,\cH_1\otimes\cH_2\mid \cH_1=\cH_2)$.  Again, without the loss of generality, we can take $t_0$ to $0$. Now, we define the compatibility of dynamical maps.
\begin{definition}
Two dynamical map $\cD_1=\{\Lambda^1_{t}\}_t\in\mathbbm{DM}(\cH,\cH)$ and $\cD_2=\{\Lambda^2_{t}\}_t\in\mathbbm{DM}(\cH,\cH)$ are said to be compatible if a joint broadcasting quantum dynamical map $\cJ=\{\Theta:\cL(\cH)\rightarrow\cL(\cH_1\otimes\cH_2)\}_t$ with $\cH_1=\cH_2=\cH$ exists such that
\begin{equation}
\cD_1=\tr_{\cH_2}\cJ;~\cD_2=\tr_{\cH_1}\cJ.
\end{equation}
where $\tr_{\cH_i}\cJ:=\{\tr_{\cH_i}\Theta\}_t$ for $i\in\{1,2\}$ \label{Def:Comp-QDM}.
\end{definition}

Definition \ref{Def:Comp-QDM} is similar to Definition \ref{Def:Comp-QC}, but now compatibility relations should hold for all $t$.

Clearly, the implementation of $\cJ$ is equivalent to the simultaneous implementation of $\cD_1$ and $\cD_2$. Note that the set of all compatible dynamical maps is convex. Therefore, a measure of incompatibility of quantum dynamical maps is the incompatibility robustness of quantum dynamical maps that we define below. The incompatibility robustness of two quantum dynamical maps $\cD_1=\{\Lambda^1_t\}_t\in\mathbbm{DM}(\cH,\cH)$ and $\cD_2=\{\Lambda^2_t\}_t\in\mathbbm{DM}(\cH,\cH)$ can be defined as:
\begin{equation}\label{Def:ROI-dmaps}
\begin{split}
R_{D}(\cD_1,\cD_2)=\min~  & \left. r\right. \\
\text{s.t.} ~ & \left. \frac{\cD_1+r\tilde{\cD}_1}{1+r}=\tr_{\cH_2}{\cJ}\right. \\
& \left. \frac{\cD_2+r\tilde{\cD}_2}{1+r}=\tr_{\cH_1}{\cJ}\right. \\
& \left. \cJ \in\mathbbm{BD}(\cH,\cH_1\otimes\cH_2\mid \cH_1=\cH_2)\right. \\
& \left. \tilde{\cD}_i\in\mathbbm{DM}(\cH,\cH) \quad i = 1,2. \right.
\end{split}
\end{equation}

Here $\tilde{\cD}_i$s are arbitrary noise dynamical maps, and the optimization is over all variables, other than the given pair of dynamical maps $(\cD_1,\cD_2)$.

Let us now, present  the following Lemma.

\begin{lemma}\label{App-lemma}
Consider two quantum channels $\Phi_1\in\mathbbm{Ch}(\cH,\cK_1)$ and $\Phi_2\in\mathbbm{Ch}(\cH,\cK_2)$ with $R_C(\Phi_1,\Phi_2)=l$. Then for all $\epsilon\geq 0$, there exists two quantum channels $\tilde{\Phi}_1\in\mathbbm{Ch}(\cH,\cK_1)$ and $\tilde{\Phi}_2\in\mathbbm{Ch}(\cH,\cK_2)$ such that the quantum channels $\frac{\Phi_1+l^{\prime}\tilde{\Phi}_1}{1+l^{\prime}}$ and $\frac{\Phi_2+l^{\prime}\tilde{\Phi}_2}{1+l^{\prime}}$ are compatible where $l^{\prime}=l+\epsilon \geq l$.
\end{lemma}

\begin{proof}
As $R_C(\Phi_1,\Phi_2)=l$, there exist $\hat{\Phi}_1\in\mathbbm{Ch}(\cH,\cK_1)$, $\hat{\Phi}_2\in\mathbbm{Ch}(\cH,\cK_2)$ and $\Phi\in\mathbbm{Ch}(\cH,\cK_1\otimes\cK_2)$ such that 

\begin{align}
\bar{\Phi}_1&=\frac{\Phi_1+l\hat{\Phi}_1}{1+l}=\tr_{\cK_2}\Phi\\
\bar{\Phi}_2&=\frac{\Phi_2+l\hat{\Phi}_2}{1+l}=\tr_{\cK_1}\Phi .
\end{align}

Now consider two completely depolarizing channels (i.e., channels with fixed output states) $\hat{\Upsilon}_{\eta_1}\in\mathbbm{Ch}(\cH,\cK_1)$ and $\hat{\Upsilon}_{\eta_2}\in\mathbbm{Ch}(\cH,\cK_2)$ (where $\eta_1\in\cS(\cK_1)$ and $\eta_2\in\cS(\cK_2)$) such that for all $\rho
\in\cS(\cH)$
\begin{align}
\hat{\Upsilon}_{\eta_1}(\rho)&=\eta_1\\
\hat{\Upsilon}_{\eta_2}(\rho)&=\eta_2.
\end{align}

In Ref. \cite[Proposition 10]{Heino-incomp-chan}, it is proved that completely depolarizing channels are compatible with any quantum channel.  Therefore, $\hat{\Upsilon}_{\eta_1}$ and $\hat{\Upsilon}_{\eta_2}$ are compatible.  The corresponding joint channel is $\hat{\Upsilon}_{\eta_1\otimes\eta_2}\in\mathbbm{Ch}(\cH,\cK_1\otimes\cK_2)$ such that for all $\rho\in\cS(\cH)$

\begin{align}
\hat{\Upsilon}_{\eta_1\otimes\eta_2}(\rho)&=\eta_1\otimes\eta_2.
\end{align}
Let $\gamma=\frac{\epsilon}{1+l}$. Clearly, $\gamma\geq 0$.

Now, consider the quantum channels
\begin{align}
\Gamma_1&=\frac{\bar{\Phi}_1+\gamma\hat{\Upsilon}_{\eta_1}}{1+\gamma}=\tr_{\cK_2}\left[\frac{\Phi+\gamma\hat{\Upsilon}_{\eta_1\otimes\eta_2}}{1+\gamma}\right]\\
\Gamma_2&=\frac{\bar{\Phi}_2+\gamma\hat{\Upsilon}_{\eta_2}}{1+\gamma}=\tr_{\cK_1}\left[\frac{\Phi+\gamma\hat{\Upsilon}_{\eta_1\otimes\eta_2}}{1+\gamma}\right].
\end{align}
Clearly, $\Gamma_1$ and $\Gamma_2$ are compatible. Recall that $l^{\prime}=l+\epsilon$. Now, it can be easily shown that
\begin{align}
\Gamma_1&=\frac{\Phi_1+l^{\prime}\tilde{\Phi}_1}{1+l^{\prime}}\\
\Gamma_2&=\frac{\Phi_2+l^{\prime}\tilde{\Phi}_2}{1+l^{\prime}}
\end{align}
where $\tilde{\Phi}_1=\frac{l\hat{\Phi}_1+\epsilon\hat{\Upsilon}_{\eta_1}}{l+\epsilon}$ and $\tilde{\Phi}_2=\frac{l\hat{\Phi}_2+\epsilon\hat{\Upsilon}_{\eta_2}}{l+\epsilon}$ are valid quantum channels. Hence, the lemma is proved.

\end{proof}

Now, defining a quantity $R^{\max}_C(\cD_1,\cD_2):=\max_tR_C(\Lambda^1_t,\Lambda^2_t)$,  we state the following result.

\begin{theorem}
For two arbitrary quantum dynamical maps $\cD_1=\{\Lambda^1_t\}_t\in\mathbbm{DM}(\cH,\cH)$ and $\cD_2=\{\Lambda^2_t\}_t\in\mathbbm{DM}(\cH,\cH)$, the equality $R_D(\cD_1,\cD_2)= R^{\max}_C(\cD_1,\cD_2)$ holds. \label{Th:RMC=RD}
\end{theorem}

\begin{proof}
 Suppose that $R_D(\cD_1,\cD_2)=l$. Therefore, from Definition \ref{Def:Comp-QDM} and the definition of incompatibility robustness of dynamical maps, there exist two dynamical maps $\bar{\cD}_1=\{\bar{\Lambda}^1_t\}_t$ and $\bar{\cD}_2=\{\bar{\Lambda}^2_t\}_t$ such that the channels $\frac{\Lambda^1_t+l\bar{\Lambda}^1_t}{1+l}$ and $\frac{\Lambda^1_t+l\bar{\Lambda}^1_t}{1+l}$ are compatible for all $t$ \label{Item:state-all-t-comp}. We know that $R^{\max}_C(\cD_1,\cD_2):=\max_tR_C(\Lambda^1_t,\Lambda^2_t)$ and suppose the maximum occurs for $t=t^{\prime}$. Therefore, $R^{\max}_C(\cD_1,\cD_2)=R_C(\Lambda^1_{t^{\prime}},\Lambda^2_{t^{\prime}})=h$ (say). Now, as discussed above, the channels $\frac{\Lambda^1_{t^{\prime}}+l\bar{\Lambda}^1_{t^{\prime}}}{1+l}$ and $\frac{\Lambda^1_{t^{\prime}}+l\bar{\Lambda}^1_{t^{\prime}}}{1+l}$ are compatible\label{Item:state-t-prime-comp}. Then, from the definition of incompatibility robustness for quantum channels, we get
\begin{equation}
R_D(\cD_1,\cD_2)=l\geq R_C(\Lambda^1_{t^{\prime}},\Lambda^2_{t^{\prime}})=R^{\max}_C(\cD_1,\cD_2) \label{Eq:RD>RMC}.
\end{equation}

 Now, suppose $R_C(\Lambda^1_t,\Lambda^2_t)=k_t$ for an arbitrary $t$. Then $k_t\leq h$. Next, from Lemma \ref{App-lemma}, it follows that there exist quantum channels $\tilde{\Lambda}^1_t$ and $\tilde{\Lambda}^2_t$ such that such that $\frac{\Lambda^1_t+h\tilde{\Lambda}^1_t}{1+h}$ and $\frac{\Lambda^2_t+h\tilde{\Lambda}^2_t}{1+h}$ are compatible for all $t$. Therefore, from the definition of incompatibility robustness of the dynamical maps, one obtains
\begin{equation}
R_D(\cD_1,\cD_2)\leq h=R^{\max}_C(\cD_1,\cD_2)\label{Eq:RD<RMC}.
\end{equation}
 Therefore, from inequalities \eqref{Eq:RD>RMC} and \eqref{Eq:RD<RMC}, it
 follows that
\begin{equation}
R_D(\cD_1,\cD_2)= h=R^{\max}_C(\cD_1,\cD_2).
\label{Eq:RD=RMC}
\end{equation}
\end{proof}

From Theorem \ref{Th:rel-comp-chan-CP-div-DM} and Theorem \ref{Th:RMC=RD}, one
can obtain the following result.
\begin{corollary}
For two arbitrary CP-divisible quantum dynamical maps $\cD_1=\{\Lambda^1_t\}_t\in\mathbbm{DM}(\cH,\cH)$ and $\cD_2=\{\Lambda^2_t\}_t\in\mathbbm{DM}(\cH,\cH)$, the equality $R_D(\cD_1,\cD_2)=R_C(\Lambda^1_0,\Lambda^2_0)$ holds.
\end{corollary}
Clearly, this corollary relates CP-indivisibility with incompatibility of channels.

\section{Illustration of incompatibility robustness of quantum channels for several dynamical maps} \label{sec:Behav-RoI-chan-examp-DM}

In this Section, we study the behaviour of incompatibility robustness of quantum channels for some specific dynamical maps, for both CP-divisible and CP-indivisible regime. Our goal is to study if the information backflow induced by CP-indivisibility, can be witnessed from non-monotonic behaviour of incompatibility robustness of quantum channels with respect to time.

To obtain the incompatibility robustness for two quantum channels $\Lambda_1(t)$ and $\Lambda_2(t)$, we implement the following algorithm, incorporating semi-definite optimization techniques. For the examples we consider here, the input Hilbert space ($\cH_{in}$) and the output Hilbert spaces ($\cH_{i,out}$) of both channels are the same as $\mathbbm{C}^d$ (i.e., $d$-dimensional complex Hilbert space for finite $d$), where $d=2$.

\begin{enumerate}[I.]
    \item\label{step1} Fix a value of time $t=0$, and we obtain the Choi matrices, $\cC_{\Lambda_1(t=0)}, \cC_{\Lambda_2(t=0)}$, of the channels, where $\cC_{\Lambda_i(t=0)} \in \mathcal{L}(\cH_{in}\otimes\cH_{i,out})$. We denote  $\cC_{\Lambda_i(t=0)}$ as $\cC_{\Lambda_i}$.
    
    \item\label{step2} We start from the value of $r=0$, and execute the following optimization:
    \begin{eqnarray}\label{eq:alg}
      &&   \max q  \nonumber \\ 
    &&  \cC_{\overline{\Lambda}_1} \in \mathcal{L}(\cH_{in}\otimes\cH_{1,out}), \ \cC_{\overline{\Lambda}_1} \geqslant 0 , \ \tr_{\cH_{1,out}}(\cC_{\overline{\Lambda}_1})=\Id_{\cH_{in}}, \nonumber \\
    &&  \cC_{\overline{\Lambda}_2} \in \mathcal{L}(\cH_{in}\otimes\cH_{2,out}), \ \cC_{\overline{\Lambda}_2} \geqslant 0 ,\ \tr_{\cH_{2,out}}(\cC_{\overline{\Lambda}_2})=\Id_{\cH_{in}}, \nonumber \\
    &&  \cC_{\Gamma} \in \mathcal{L}(\cH_{in}\otimes\cH_{1,out}\otimes\cH_{2,out}), \ \cC_{\Gamma} \geqslant q\mathbbm{1}_{\cH_{in}\otimes\cH_{1,out}\otimes\cH_{2,out}} ,
    \nonumber \\
    && \tr_{\cH_{2,out}}(\cC_{\Gamma}) = \frac{\cC_{\Lambda_1} + r  \cC_{\overline{\Lambda}_1}}{1+r}, \nonumber \\
    &&\tr_{\cH_{1,out}}(\cC_{\Gamma}) = \frac{\cC_{\Lambda_2} + r  \cC_{\overline{\Lambda}_2}}{1+r} . 
    \end{eqnarray}
    \begin{itemize}
        \item If the value of $q$ is negative, it indicates that $\Lambda_1(t=0)$ and $\Lambda_2(t=0)$ are incompatible for that specific value of $r$. In this case, we proceed by repeating step \eqref{step2} with an updated value of $r$, i.e., $r = r + \delta r$.
        \item If the value of $q$ is greater than or equal to zero, it signifies that $\Lambda_1(t=0)$ and $\Lambda_2(t=0)$ have become compatible. We store the current value of $r$ as the incompatibility robustness. Subsequently, we return to step \eqref{step1} and increment the parameter $t$ by $\delta t$.
    \end{itemize}
\end{enumerate}
The incompatibility robustness with respect to completely depolarising (CD) noise can be determined using the same method. In this case, instead of the constraint $\cC_{\overline{\Lambda}_i} \in \mathcal{L}(\cH_{in}\otimes\cH_{i,out}),~\cC_{\overline{\Lambda}_i} \geqslant 0 ,~\tr_{\cH_{i,out}}(\cC_{\overline{\Lambda}_i})=\Id_{\cH_{in}}$, the constraint $\cC_{\overline{\Lambda}_i}=\Id_{\cH_{in}}\otimes\eta_i,\eta_i\in\cL(\cH_{i,out})~,~\eta_i \geqslant 0,~\tr[\eta_i]=1$ needs to be imposed for each $i=1,2$. We refer the readers to the Ref. \cite{github_link} for the codes to generate some of plots that have been discussed below. The codes for the other plots are similar. 

\subsection{Qubit depolarising dynamical maps} 

Consider a qubit depolarizing dynamical map $\cD=\{\Lambda_t\}_t$ of the form
\begin{equation}
    \Lambda_t(\rho)=w(t)\rho+(1-w(t))\frac{\Id}{2}.
\end{equation}
Here, $\Id=\Id_{2\times 2}$. The Choi matrix of $\Lambda_t$ can be written as

\begin{equation}
    \cC_{\Lambda_t}=\begin{bmatrix}
        \frac{1+w(t)}{2}&0&0&w(t)\\
        0&\frac{1-w(t)}{2}&0&0\\
        0&0&\frac{1-w(t)}{2}&0\\
        w(t)&0&0&\frac{1+w(t)}{2}
    \end{bmatrix}.
\end{equation}

Now, let us first consider the CP-divisible scenario, where we can take $w(t)=e^{-\lambda t}$ with $\lambda$ to be some positive real constant and we call the dynamical map $\cD_1=\{\Lambda^1_t\}_t$ (i.e., $\Lambda^1_t(\rho)=e^{-\lambda t}\rho+(1-e^{-\lambda t})\frac{\Id}{2}$). This is a divisible depolarising dynamical map, which can be shown very easily \citep{RHP}.  (For our purpose of study, we take $\lambda$ to be $0.5$. The step size $\delta r$ is taken to be $0.005$. The values of $t$ are taken from $0$ to $1$ (in units of $1/\lambda$), and  the interval of $t$ is taken to be $0.01$ in all the cases).

Note that it can be directly shown from the lindblad evolution of trace distance that the necessary condition for non monotonic behaviour of the trace distance is the breaking down of divisibility \cite{blp1}. The increment of trace distance between any two possible states of a system (evolving through a dynamical map) w.r.t. time is an indication of information backflow from the environment to the system.
    In the following examples studied by us, we use the trace distance curve to show that non-monotonicity of it has similarity with non-monotonicity of the incompatibility robustness curve.
The incompatibility robustness of two copies of $\Lambda^1_t$ (i.e., $R_C(\Lambda^1_t,\Lambda^1_t)$ and $R^{\cC\cD}_{C}(\Lambda^1_t,\Lambda^1_t)$ is plotted w.r.t. time $t$ in Fig \ref{Fig:CP-div-depolarise} and we observe the monotonic behaviour.  Incompatibility of channels becomes permanently zero at approximately $t=0.81$.

\begin{figure}
    \centering
    \includegraphics[scale=0.43]{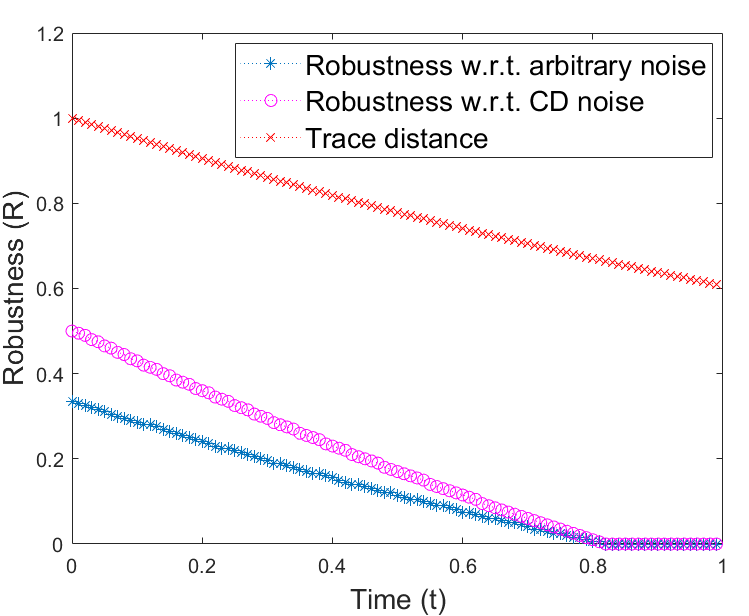}
    \caption{This plot shows the behaviour of incompatibility robustness of quantum channels (w.r.t. completely depolarising (CD) noise and generic noise) for two copies of the depolarising dynamical map $\cD_1$ (with $w(t)=e^{-\lambda t}$ where $\lambda=0.5$) w.r.t. time $t$. It also shows the trace distance between $\Lambda^1_t(\ket{0}\bra{0})$ and $\Lambda^1_t(\ket{1}\bra{1})$. Clearly, there is no information backflow and the behaviour of incompatibility robustness is monotonically non-increasing.}
    \label{Fig:CP-div-depolarise}
\end{figure}

Now, let us take $w(t)=e^{-\lambda t}\cos^2\omega t$, for the example of a CP-indivisible dynamical map. The inclusion of the cosine function imparts  oscillation in the term, allowing information backflow from environment to the system. For our purpose, we take $\lambda = 0.5, \omega = 5 \pi$ and denote the dynamical map as $\cD_2=\{\Lambda^2_t\}_t$ (i.e., $\Lambda^2_t(\rho)=e^{-0.5 t}\cos^2 5 \pi t\rho+(1-e^{-0.5 t}\cos^2 5 \pi t)\frac{\Id}{2}$). In this case, the incompatibility robustness of two copies of $\Lambda^2_t$ (i.e., $R_C(\Lambda^2_t,\Lambda^2_t)$ and $R^{\cC\cD}_{C}(\Lambda^2_t,\Lambda^2_t)$) is plotted w.r.t. time $t$ in Fig \ref{Fig:CP-in-div-depolarise}.  Note that the behaviour of incompatibility of channels is non-monotonic w.r.t time $t$, but permanently becomes zero approximately at $t=0.81$.

\begin{figure}
    \centering
    \includegraphics[scale=0.43]{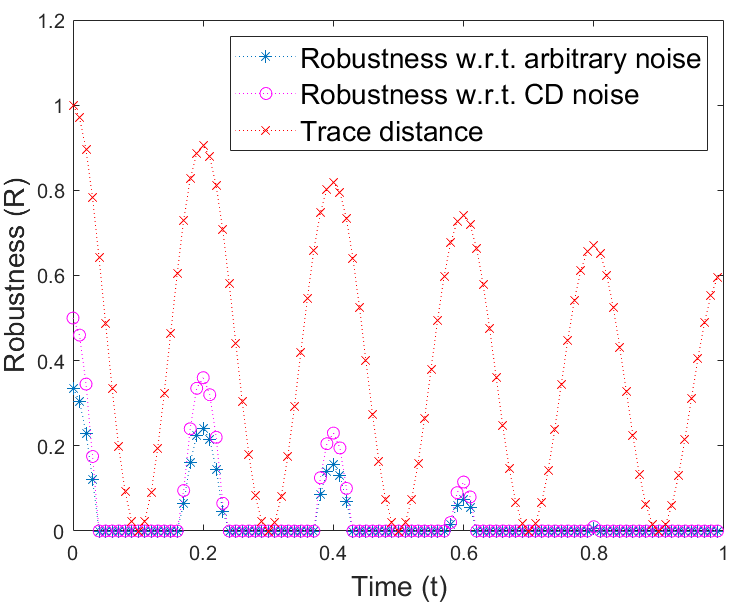}
    \caption{This plot shows the behaviour of incompatibility robustness of quantum channels (w.r.t. completely depolarising (CD) noise and generic noise) for two copies of the depolarising dynamical map $\cD_2$ with $w(t)=e^{-\lambda t}\cos^2\omega t$ where $\lambda=0.5, \omega = 5\pi$  w.r.t. time $t$. It also shows the trace distance between $\Lambda^2_t(\ket{0}\bra{0})$ and $\Lambda^2_t(\ket{1}\bra{1})$. Clearly, there is information backflow that can be witnessed from non-monotonicity of trace distance, and also, the behaviour of incompatibility robustness is non-monotonic.}
    \label{Fig:CP-in-div-depolarise}
\end{figure}

We know that a CP-divisible map does not exhibit information backflow. Therefore, to demonstrate the connection between non-monotonic behaviour of incompatibility robustness and information backflow, we plot incompatibility robustness of quantum channels for the dynamcal maps $\cD_1$ (CP-divisible) and $\cD_2$ (CP-indivisible) i.e., we plot $R_C(\Lambda^1_t,\Lambda^2_t)$ and $R^{\cC\cD}_{C}(\Lambda^1_t,\Lambda^2_t)$ for time $t$ in Fig. \ref{Fig:CP-div-CP-in-div-depolarise}. Here, $\cD_1$ can be considered as a \emph{reference CP-divisible dynamical map} and information backflow of $\cD_2$ can be witnessed through the non-monotonic behaviour of $R_C(\Lambda^1_t,\Lambda^2_t)$ and $R^{\cC\cD}_{C}(\Lambda^1_t,\Lambda^2_t)$ w.r.t. $t$. This is one of the plots that help us to define a measure of CP-indivisibility based on incompatibility of channels in Section \ref{sec:cp-in-meas}.

\begin{figure}
    \centering
    \includegraphics[scale=0.43]{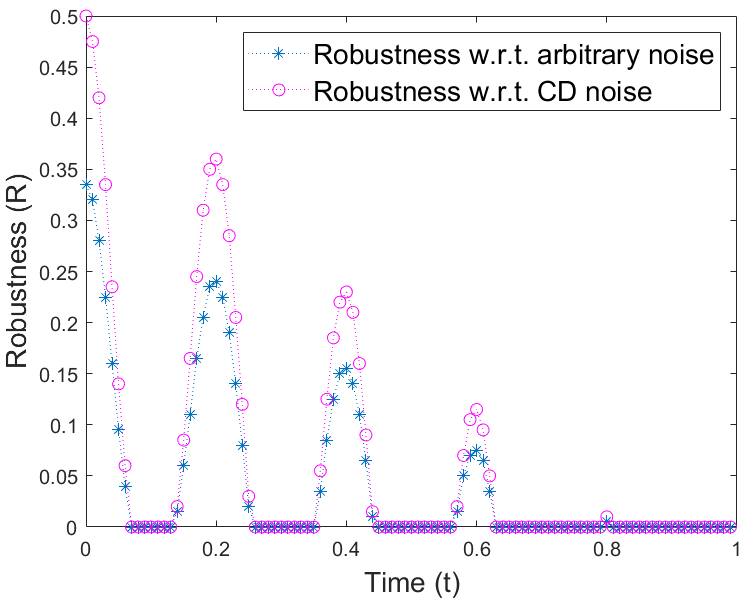}
    \caption{This plot shows the non-monotonic behaviour of incompatibility robustness of channels for dynamical maps $\cD_1$ and $\cD_2$ (i.e., $R_C(\Lambda^1_t,\Lambda^2_t)$ and $R^{\cC\cD}_{C}(\Lambda^1_t,\Lambda^2_t)$) w.r.t. time $t$. Although $\cD_1$ is CP-divisible, $\cD_2$ is CP-indivisible and exhibits information backflow which is the cause of non-monotonic behaviour of incompatibility robustness.}
    \label{Fig:CP-div-CP-in-div-depolarise}
\end{figure}

From Fig. \ref{Fig:CP-in-div-depolarise} and Fig. \ref{Fig:CP-div-CP-in-div-depolarise}, we observe that there are some values of $t$ (between two ripples) where incompatibility robustness remains zero (i.e., non-increasing), but there exists non-monotonic behaviour in trace distance that indicates the information backflow. Therefore, the information backflow can not be witnessed from the graph of the incompatibility robustness of quantum channels for all those values of $t$. But it is possible to eliminate such a limitation if we carefully choose the reference CP-divisible dynamical map. For example, consider $\cD_I=\{\mathbbm{I}_t\}_t$ be the identity dynamical map i.e., $\mathbbm{I}_t=\mathbbm{I}_{\cH_q}$ is an identity channel for all at $t$ where $\cH_q$ is the qubit Hilbert space. Clearly, $\cD_I$ is a CP divisible dynamical map. Let's choose $\cD_I$ as the reference CP-divisible dynamical map and plot the incompatibility robustness of channels for dynamical maps $\cD_I$ and $\cD_2$ (i.e., $R_C(\mathbbm{I}_t,\Lambda^2_t)$ and $R^{\cC\cD}_{C}(\mathbbm{I}_t,\Lambda^2_t)$) w.r.t. time $t$ in Fig. \ref{Fig:id-CP-in-div-depolarise}. From Fig. \ref{Fig:id-CP-in-div-depolarise}, we observe that for all values of $t$ (although, displayed only for finite range of time), there is non-monotonic behaviour of both trace distance  and incompatibility robustness of channels and for an arbitrary time $t$, if the information backflow is non-zero (observed from the trace distance graph) then the incompatibility robustness of quantum channels is strictly increasing. Therefore, for an arbitrary time $t$, if the information backflow is non-zero then it can be witnessed from the graph of incompatibility of quantum channels. Therefore, above-said limitation has been removed by choosing $D_I$ as the reference CP-divisible dynamical map.

\begin{figure}
    \centering
    \includegraphics[scale=0.43]{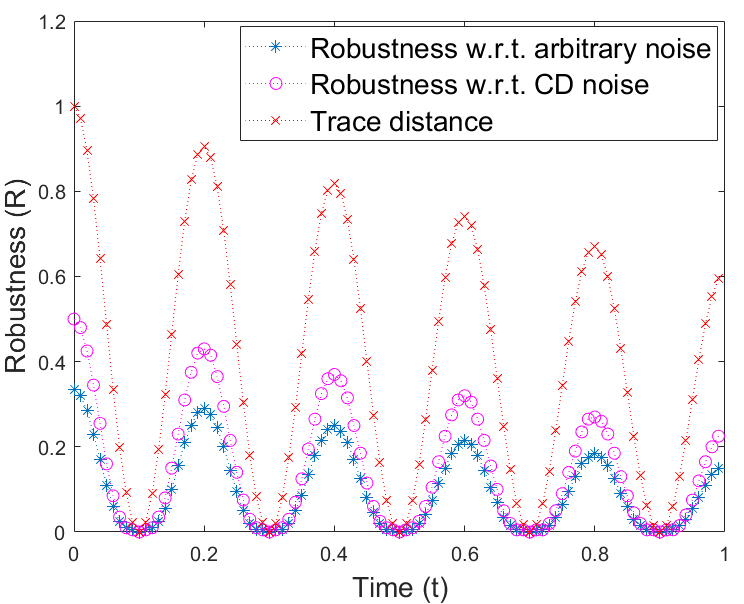}
    \caption{This plot shows the non-monotonic behaviour of incompatibility robustness of channels for dynamical maps $\cD_I$ and $\cD_2$ (i.e., $R_C(\mathbbm{I}_t,\Lambda^2_t)$ and $R^{\cC\cD}_{C}(\mathbbm{I}_t,\Lambda^2_t)$) w.r.t. time $t$. It also shows the trace distance between $\Lambda^2_t(\ket{0}\bra{0})$ and $\Lambda^2_t(\ket{1}\bra{1})$. Clearly, for an arbitrary time $t$ if the information backflow (observed from the trace distance graph) is non zero then the incompatibility robustness of quantum channels is also strictly increasing. Therefore, for an arbitrary time $t$, if the information backflow is non-zero, it can be witnessed from the graph of incompatibility of quantum channels.}
    \label{Fig:id-CP-in-div-depolarise}
\end{figure}

\subsection{Qubit amplitude damping dynamical maps}
In this section, we consider qubit amplitude damping dynamical map $\cD_{ad}=\{\Gamma^{ad}_t\}_t$ where the Choi matrix of $\Gamma^{ad}_t$ can be written as

\begin{equation}
    \cC_{\Gamma^{ad}_t}=\begin{bmatrix}
        1&0&0&\sqrt{1-w(t)}\\
        0&0&0&0\\
        0&0&1-w(t)&0\\
        \sqrt{1-w(t)}&0&0&w(t)
    \end{bmatrix}.
\end{equation}
where $0\leq w(t)\leq 1$. We choose $w(t)$ in such a way that dynamics is CP-indivisible and exhibits information backflow. Taking $w(t)=1-e^{-\alpha t}cos^2\omega t$, we set the value of $\alpha=0.5$ and $\omega=5\pi$. As discussed in previous section (mainly from Fig. \ref{Fig:id-CP-in-div-depolarise}), we observed that the identity dynamical map $\cD_I$ is possibly a suitable reference CP-divisible dynamical map. Therefore, we plot the incompatibility robustness of channels for dynamical maps $\cD_I$ and $\cD_{ad}$ (i.e., $R_C(\mathbbm{I}_t,\Gamma^{ad}_t)$ and $R^{\cC\cD}_{C}(\mathbbm{I}_t,\Gamma^{ad}_t)$) w.r.t. time $t$ in Fig. \ref{Fig:id-CP-in-div-amp-damp} and we observe information backflow as well as non-monotonic behaviour of incompatibility robustness of channels.

\begin{figure}
    \centering
    \includegraphics[scale=0.43]{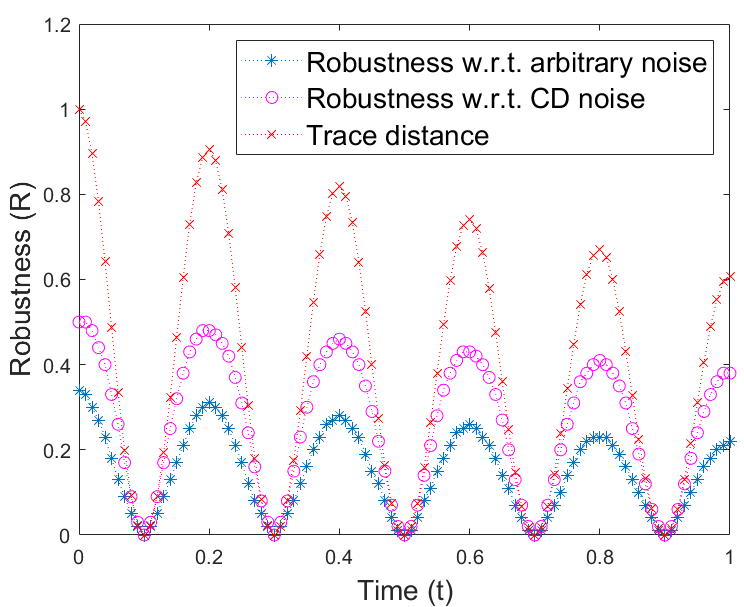}
    \caption{This plot shows the non-monotonic behaviour of incompatibility robustness of channels for dynamical maps $\cD_I$ and $\cD_{ad}$ (i.e., $R_C(\mathbbm{I}_t,\Gamma^{ad}_t)$ and $R^{\cC\cD}_{C}(\mathbbm{I}_t,\Gamma^{ad}_t)$) w.r.t. time $t$. It also shows the trace distance between $\Gamma^{ad}_t(\ket{0}\bra{0})$ and $\Gamma^{ad}_t(\ket{1}\bra{1})$. Clearly, for an arbitrary time $t$ if the information backflow (witnessed from non-monotonicity of trace distance) is non zero then the incompatibility robustness of quantum channels is also strictly increasing and therefore, if the information backflow is non-zero for an any time $t$, it can be witnessed from the graph of incompatibility of quantum channels.}
    \label{Fig:id-CP-in-div-amp-damp}
\end{figure}

 From the above discussion through examples of depolarizing as well as amplitude damping channels (mainly from Fig. \ref{Fig:CP-div-CP-in-div-depolarise}, Fig. \ref{Fig:id-CP-in-div-depolarise} and Fig. \ref{Fig:id-CP-in-div-amp-damp}), we observe the \emph{simultaneous presence} of both information backflow (a signature of CP-indivisibility) and non-monotonicity of incompatibility robustness of channels. Therefore, this observation  motivates us to define a measure of CP-indivisibility based on incompatibility robustness of channels. We will define such a measure in Sec. \ref{sec:cp-in-meas}.

 \subsection{CP-indivisible dynamical maps without information backflow}

 Although information backflow which may be measured using trace distance is a signature of CP-indivisibility, it is not equivalent to CP-indivisibility \cite{Rivas_P_divisible, Hall_Eternal-non-Markovian}. In order to illustrate this point, let us consider the Choi matrix of a dynamical map $\cD_{et}=\{\Lambda^{et}_t\}_t$, given by 
 \begin{equation}
    \cC_{\Lambda^{et}_t}=\begin{bmatrix}
        A(t)&0&0&B(t)\\
        0&1-A(t)&0&0\\
        0&0&1-A(t)&0\\
        B(t)&0&0&A(t)
    \end{bmatrix}.
\end{equation}
where $A(t)=\frac{1+e^{-2t}}{2}$ and $B(t)=e^{-\int^t_0(1-tanh x)dx}$. Such a dynamical map is CP-indivisible, but does not show information backflow (i.e., the trace distance is monotonically non-increasing w.r.t. time $t$) \cite{Hall_Eternal-non-Markovian}. 

\begin{figure}
    \centering
    \includegraphics[scale=0.43]{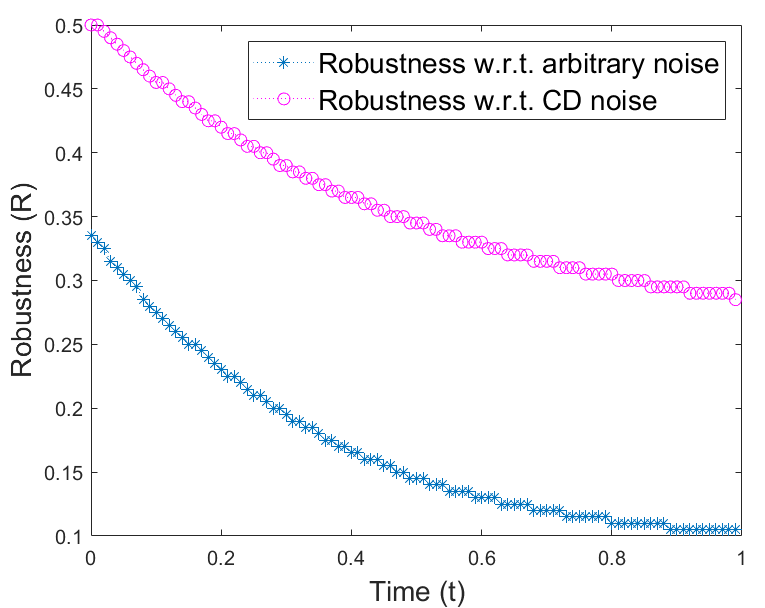}
    \caption{This plot shows monotonic behaviour of incompatibility robustness of channels for dynamical maps $\cD_I$ and $\cD_{et}$ (i.e., $R_C(\mathbbm{I}_t,\Lambda^{et}_t)$ and $R^{\cC\cD}_{C}(\mathbbm{I}_t,\Lambda^{et}_t)$) w.r.t. time $t$.}
    \label{Fig:id-et-dm}
\end{figure}

For the above dynamics we plot the robustness with respect to both arbitrary and depolarizing noise in Fig. \ref{Fig:id-et-dm}. No non-monotonicity indicative of information backflow is displayed.  Although, we did not find any non-monotonic behaviour of incompatibility robustness of quantum channels w.r.t. time $t$, it is a matter of further investigation to conclude whether it is possible for non-monotonic behaviour to
be revealed for any other choice of the reference CP-divisible map, instead of the identity dynamical map.

\section{CP-indivisibility and incompatibility of channels as
resource for quantum teleportation}\label{sec:telport-cp-in-divisible-incomp}

Quantum teleportation is a very well-known quantum communicaion protocol \cite{bennett-teleport,Popescu_teleport}. Although perfect quantum teleportation can be performed using maximally entangled states, it can be performed imperfectly with mixed non-maximal entangled states, in general. A criterion for a two-qubit state to be useful in quantum teleportation is provided in Ref. \cite{horodecki_teleportation}. Consider a two-qubit state $\rho_{AB}$. Define a $3\times 3$ matrix $S_{\rho_{AB}}=[s(\rho_{AB})_{ij}]$ such that $s(\rho_{AB})_{ij}=\tr[\rho_{AB}(\sigma_i\otimes\sigma_j)]~\forall i,j\in\{x,y,z\}$ where $\sigma_j$s are Pauli matrices. Let, $N(\rho_{AB})=Tr[\sqrt{S_{\rho_{AB}}^{\dagger}S_{\rho_{AB}}}]$. A quantum state is useful for  teleportation  iff $N(\rho_{AB})>1$ and in this case, the maximum fidelity is $F_{max} \equiv \frac{1}{2}[1+\frac{1}{3}N(\rho)] > 2/3 $\cite{horodecki_teleportation}.

Now, consider the two-qubit maximally entangled state $\ket{\Psi_{-}}=\frac{1}{\sqrt{2}}[\ket{01}-\ket{10}]$ where $\{\ket{0}, \ket{1}\}$ are the eigen basis of $\sigma_z$. Let us take the dynamical map $\cD_2=\{\Lambda^2_t\}_t$ from Section \ref{sec:Behav-RoI-chan-examp-DM}. On application of $\Lambda^2_t$ on the state $\ket{\Psi_{-}}\bra{\Psi_{-}}$, we obtain
\begin{align}
    \rho^{\prime}_{AB}&=(\cI_A\otimes\Lambda^2_t)(\rho_{AB})\nonumber\\
    &=w(t)\ket{\Psi_{-}}\bra{\Psi_{-}}+(1-w(t))\frac{\Id_{4\times 4}}{4}
\end{align}
where $w(t)=e^{-\lambda t}cos^2\omega t$ with $\lambda=0.5$ and $\omega=5\pi$. Now,
 \begin{align}
     S_{\rho^{\prime}_{AB}}=\begin{bmatrix}
         -w(t)& 0& 0\\
         0& -w(t)& 0\\
         0& 0& -w(t)
     \end{bmatrix}.
 \end{align}
It follows that $N(\rho^{\prime}_{AB})=3w(t)$ and therefore, $F_{max}=\frac{1}{2}[1+\frac{1}{3}w(t)]$ for $N(\rho^{\prime}_{AB})>1$ and otherwise, $F_{max}=\frac{2}{3}$.

\begin{figure}
    \centering
    \includegraphics[scale=0.43]{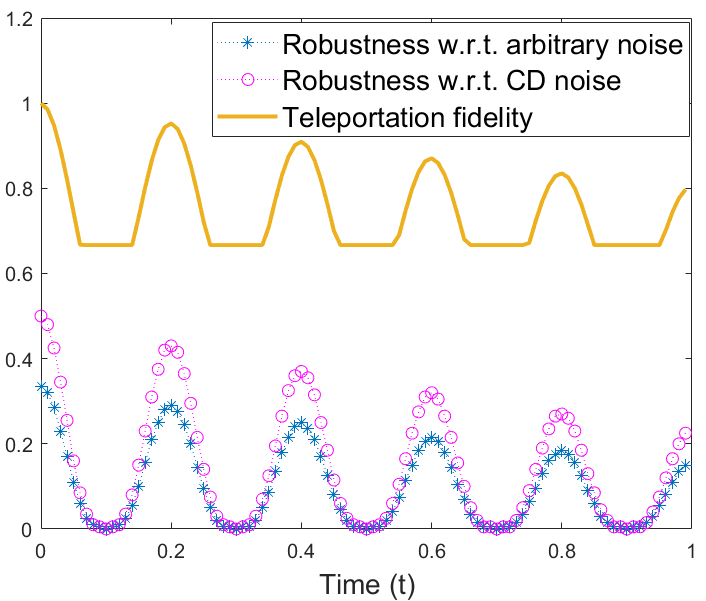}
    \caption{This plot shows the maximum fidelity for teleportation using the state $\rho^{\prime}_{AB}=(\cI_A\otimes\Lambda^2_t)(\rho_{AB})$ and the non-monotonic behaviour of incompatibility robustness of channels for dynamical maps $\cD_I$ and $\cD_{2}$ (i.e., $R_C(\mathbbm{I}_t,\Lambda^{2}_t)$ and $R^{\cC\cD}_{C}(\mathbbm{I}_t,\Lambda^{2}_t)$) w.r.t. time $t$. Clearly, the teleportation fidelity increases with a corresponding increase in incompatibility robustness.}
    \label{Fig:teleport-incomp}
\end{figure}

In Fig. \ref{Fig:teleport-incomp} we plot the robustness measure and teleportation fidelity versus time.
From the figure we observe that the teleportation fidelity rises with increase of the incompatibility robustness  which is a signature of CP-indivisibility. 
This clearly establishes CP-indivisibility is a resource for quantum teleportation.
 
\section{Measuring CP-indivisibility using incompatibility robustness of quantum channels} \label{sec:cp-in-meas}
In this section, we propose a measure of CP-indivisibility based on incompatibility robustness of quantum channels. As we can see from the study in Sec. \ref{sec:Behav-RoI-chan-examp-DM}, the incompatibility robustness of channels shows the signature of CP-indivisibility induced information backflow through non-monotonic behaviour. Therefore, it is evident that an information backflow measure based measure of CP-indivisibility \citep{blp1} can be constructed using the incompatibility robustness of channels. Here, we follow the procedure proposed by Laine et.al \citep{blp1} to construct such a measure of CP-indivisibility.

Consider a dynamical map $\cD=\{\Lambda_t\}_t\in\mathbbm{DM}(\cH,\cH)$. We have to construct a CP-indivisibility measure of $\cD$.  For this, take an arbitrary CP-divisible dynamical map $\bar{\cD}=\{\bar{\Lambda}_t\}_t\in\mathbbm{DM}(\cH,\cH)$ that acts as a reference CP-divisible dynamical map. Let $\theta(t):=\frac{d R_C(\Lambda_t,\bar{\Lambda}_t)}{d t}$. Then, we define the CP-indivisibility measure $\cN(\cD)$ as
\begin{equation}
    \cN(\cD)=\sup_{\bar{\cD}\in\mathbbm{DM}(\cH,\cH)_{\cC\cP}}\int_{\theta(t)>0}R_C(\Lambda_t,\bar{\Lambda}_t)dt
\end{equation}
where $\mathbbm{DM}(\cH,\cH)_{\cC\cP}$ is the set of all CP-divisible dynamical map acting from $\cL(\cH)$ to $\cL(\cH)$. Clearly,  $\cN(\cD)\geq 0$.
We obtain from Theorem \ref{Th:rel-comp-chan-CP-div-DM} that this measure is always \emph{zero} for CP-divisible dynamical maps (as $\theta(t)\leq 0$ for those dynamical maps). But non-monotonic behaviour of incompatibility robustness of channels in Fig. \ref{Fig:CP-div-CP-in-div-depolarise} and Fig. \ref{Fig:id-CP-in-div-depolarise} suggests that that the proposed measure is non-zero for the CP-indivisible depolarising dynamical map $\cD_2$. Similarly, non-monotonic behaviour of incompatibility robustness of channels in Fig \ref{Fig:id-CP-in-div-amp-damp} suggests that this measure is non-zero for the CP-indivisible amplitude damping dynamical map $\cD_{ad}$. Therefore, in short, there exists CP-indivisible dynamical maps for which this measure is non-zero. It is evident that as the expression of $  \cN(\cD)$ consists of integration with respect to time $t$ over the range from $0$ to $\infty$, this measure $  \cN(\cD)$ can take an arbitrary value and it is not normalised. For the sake of compactness, we can also propose a normalised measure of the form 
\begin{equation}
    \mathfrak{N}(\cD)=\frac{  \cN(\cD)}{1+  \cN(\cD)}
\end{equation}
Clearly, for any dynamical map $\cD$, $0\leq \mathfrak{N}(\cD)\leq 1$ and $\mathfrak{N}(\cD)=0$ for CP-divisible dynamical maps (as $\cN(\cD)=0$ for CP-divisible dynamical maps) and $\mathfrak{N}(\cD)$ is \emph{non-zero} whenever $\cN(\cD)$ is non-zero.
Our proposed measure is on a similar footing of the information backflow based non-Markovianity measure  \citep{blp1}.

\section{conclusions}\label{sec:conc}

To summarize, in this work, we have considered two important features of quantum theory, {\it viz.} CP-indivisibility and incompatibility of channels, that arise naturally in several practical situations of quantum dynamics.
These two properties have been utilized as resources in several quantum information processing protocols. Our present analysis enables the characterization of CP-indivisibility of dynamical maps using incompatibility of channels. We have shown that incompatibility robustness of channels for two CP-divisible dynamical maps is monotonically non-increasing with respect to time. We have shown that for two dynamical maps and for a particular time $t$, the incompatibility robustness of quantum measurements, is upper bounded by the incompatibility robustness of quantum channel.

Furthermore, we have explicitly analyzed the case of qubit depolarising dynamical maps and qubit amplitude damping dynamical maps as examples. We have shown that in the case of the CP-divisible regime, incompatibility robustness of channels is monotonically non-increasing with respect to time. But in the case of a CP-indivisible regime, it loses its monotonic behaviour in both cases. The examples studied here clearly illustrate the  simultaneous presence of information backflow from the environment to the system, as reflected by the non-monotonic behaviour of the trace distance and non-monotonic behaviour of the incompatibility robustness of quantum channels for both generic and completely depolarising noise models. We have further shown through an example
how information backflow acts as a resource for quantum teleportation. Additionally, we have proposed a measure of CP-indivisibility based on incompatibility robustness of quantum channels.

The results obtained from our present study motivate certain directions of future research. It may be worthwhile to explore whether the incompatibility robustness of quantum channels can be used to witness CP-indivisible maps, 
such as the one studied in Section IV C, that do not show information backflow \cite{PhysRevA.89.042120}. Moreover, it would also be interesting to investigate whether the CP-indivisibility measure proposed in Sec. \ref{sec:cp-in-meas}, can be useful to quantify the performance of some specific information-theoretic or thermodynamic tasks.

{\it Acknowledgements:}

ASM acknowledges support from the DST project DST/ICPS/QuEST/2019/Q98.

\bibliography{reference}

\end{document}